\documentclass[11pt,a4paper]{article}

\usepackage[colorlinks=true,urlcolor=Blue,citecolor=Green,linkcolor=BrickRed]{hyperref}
\usepackage[usenames,dvipsnames]{xcolor}
\usepackage[utf8]{inputenc}
\usepackage{algorithm}
\usepackage{algorithmicx}
\usepackage[noend]{algpseudocode}
\usepackage{cite}
\usepackage[OT4]{fontenc}
\usepackage{amsthm}
\usepackage{amssymb}
\usepackage{amsmath}
\usepackage{amsfonts}
\usepackage{todonotes}
\usepackage{authblk}
\usepackage{fullpage}

\title{A Family of Approximation Algorithms for the Maximum Duo-Preservation String Mapping Problem}
\date{}
\author[1]{Bartłomiej Dudek}
\author[1,2]{Paweł Gawrychowski}
\author[1]{Piotr Ostropolski-Nalewaja}
\affil[1]{Institute of Computer Science, University of Wrocław, Poland}
\affil[2]{University of Haifa, Israel}

\algnewcommand\algorithmicforeach{\textbf{for each}}
\algdef{S}[FOR]{ForEach}[1]{\algorithmicforeach\ #1\ \algorithmicdo}

\newcommand{\Oh}{{O}}
\newcommand{\ALG}{ALG}
\newcommand{\OPT}{OPT}
\newcommand{\eps}{\varepsilon}
\newcommand{\CC}{\mathcal{C}}
\newcommand{\ddd}{\ldots}
\newcommand{\pair}[2]{\langle {#1},{#2}\rangle}
\newcommand\bigforall{\mbox{\Large $\mathsurround0pt\forall$}}
\newcommand{\close}{\mathsf{Close}}
\newcommand{\confl}{\mathsf{Overlap}}
\newcommand{\gaps}{\mathsf{gaps}}

\newtheorem{theorem}{Theorem}[section]
\newtheorem{lemma}[theorem]{Lemma}
\newtheorem{corollary}[theorem]{Corollary}

\begin{document}

\maketitle

\begin{abstract}  
In the Maximum Duo-Preservation String Mapping problem we are given two strings and wish
to map the letters of the former to the letters of the latter so as to maximise the number of
duos. A duo is a pair of consecutive letters that is mapped to a pair of consecutive letters
in the same order. This is complementary to the well-studied Minimum Common String Partition
problem, where the goal is to partition the former string into blocks that can be permuted
and concatenated to obtain the latter string. 

Maximum Duo-Preservation String Mapping is APX-hard. After a series of
improvements, Brubach [WABI 2016] showed a polynomial-time $3.25$-approximation algorithm.
Our main contribution is that for any $\epsilon>0$ there exists a polynomial-time
$(2+\epsilon)$-approximation algorithm. Similarly to a previous solution by Boria et al.
[CPM 2016], our algorithm uses the local search technique. However, this is used only
after a certain preliminary greedy procedure, which gives us more structure
and makes a more general local search possible.
We complement this with a specialised version of the algorithm that achieves $2.67$-approximation
in quadratic time.
 \end{abstract}
 
 \section{Introduction}

A fundamental question in computational biology and, consequently, stringology, is comparing
similarity of two strings. A textbook approach is to compute the edit distance, that is, the
smallest number of operations necessary to transform one string into another, where every
operation is inserting, removing, or replacing a character. While this can be efficiently
computed in quadratic time, a major drawback from the point of view of biological applications
is that every operation changes only a single character. Therefore, it makes sense to
also allow moving arbitrary substrings as a single operation to obtain edit distance with
moves. Such relaxation makes computing the smallest number of
operations NP-hard~\cite{Shapira2002}, but Cormode and Muthukrishnan~\cite{Cormode2002}
showed an almost linear-time $\Oh(\log n \cdot \log^{*}n)$-approximation algorithm.
The problem is already interesting if the only allowed operation is moving a substring.
This is usually called the Minimum Common String Partition (MCSP). Formally, we are
given two strings $X$ and $Y$, where $Y$ is a permutation of $X$. The goal is to cut $X$ into
the least number of pieces that can be rearranged (without reversing) and concatenated
to obtain $Y$.

MCSP is known to be APX-hard~\cite{GoldsteinHardness}. Chrobak et al.~\cite{Chrobak2005}
analysed performance of the simple greedy approximation algorithm 
that in every step extracts the longest common substring from the input strings and Kaplan and
Shafrir~\cite{Kaplan} further improved their bounds. This simple greedy algorithm can be
implemented in linear time~\cite{GoldsteinLewenstein2014} and was further tweaked to obtain better practical results~\cite{He2007}.
Also, an exact exponential
time algorithm~\cite{Fu2011} and different parameterizations were 
considered~\cite{Jiang2012,Bulteau2013,Bulteau2014,Damaschke2008}.

There was also some interest in the complementary problem called the Maximum
Duo-Preservation String Mapping (MPSM), introduced by Chen et al.~\cite{Chen2014}.
The goal there is to map the letters of $X$ to the letters
of $Y$ so as to maximise the number of preserved duos. A duo is a pair of consecutive letters, and
a duo of $X$ is said to be preserved if its pair of consecutive letters is mapped to a pair of
consecutive letters of $Y$ (in the same order).
MCSP and MPSM are indeed complementary, as one can think of preserving a duo
as not splitting its two letters apart to see that the number of preserved duos and the number
of pieces add up to $|X|$.
Of course, this says nothing about the relationship between the approximation
guarantees for both problems.
Chen et al.~\cite{Chen2014} designed a $k^{2}$-approximation algorithm based on
linear programming for the restricted version of the problem, called $k$-MPSM, where each
letter occurs at most $k$ times. This was soon followed by an APX-hardness proof
of 2-MPSM and a general 4-approximation algorithm provided by Boria et al.~\cite{Boria2014}. 
The approximation ratio was then improved to 3.5~\cite{DuaItalians} using a particularly
clean argument based on local search. Finally, Brubach~\cite{Wabi16} obtained a 3.25-approximation,
and Beretta et al.~\cite{Beretta} considered parameterized tractability.

Our main contribution is a family of polynomial-time approximation algorithms
for MPSM: for any $\eps > 0$, we show a polynomial-time $(2+\eps)$-approximation
algorithm. We complement this with a specialised (and simplified) version of the algorithm that
achieves 2.67-approximation in quadratic time, which already improves on the approximation
guarantee and the running time of the previous solutions, as the running time of the
3.5-approximation was $\Oh(n^{4})$.
At a high level, we also apply local search, that is, we iteratively try to slightly change
the current solution as long as such a change leads to an improvement. The intuition
is that being unable to find such local improvement should imply a $(2+\eps)$-approximation
guarantee. This requires considering larger and larger neighbourhoods of the current solution
for smaller and smaller $\eps$ and seems problematic already for $\eps=1$.  To overcome
this, we apply local search only after a certain preliminary greedy procedure, which gives
us more structure and makes a more general local search possible.

\section{Preliminaries}

In the Maximum Duo-Preservation String Mapping (MPSM) we are given two strings $X$ and
$Y$, where $Y$ is a permutation of $X$. The goal is to map the letters of $X$ to the letters
of $Y$ so as to maximise the number of preserved duos. A duo is a pair of consecutive letters, and
a duo of $X$ is said to be preserved if its pair of consecutive letters is mapped to a pair of
consecutive letters of $Y$ (in the same order). This can be restated by
creating a bipartite graph $G=(A\dot{\cup}B,E)$, where $n=|X|-1=|A|=|B|$ and 
$A= \{a_1,a_2,\ddd,a_n\}$ and $B= \{b_1,b_2,\ddd,b_n\}$. Node $a_{i}$ corresponds
to duo $(X[i],X[i+1])$ and similarly $b_{i}$ corresponds to $(Y[i],Y[i+1])$.
Two nodes are connected with an edge if their corresponding duos are the same,
that is, $E=\{(a_i,b_j): X[i]=Y[j] \text{ and }  X[i+1]=Y[j+1]\}$. See Figure~\ref{fig:graph}.

\begin{figure}[h]
\begin{centering}
\includegraphics{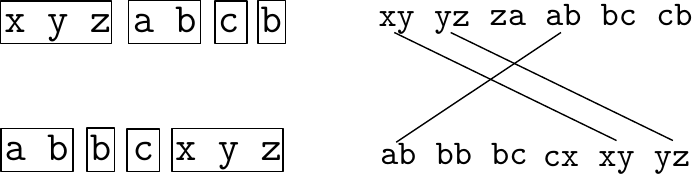}
\caption{An optimal solution of MCSP for strings \texttt{xyzabcb} and \texttt{abbcxyz} (left).
It corresponds to a solution of MPSM, where the mapping preserves duos $(x,y),(y,z)$, and $(a,b)$ (right).}
\label{fig:graph}
\end{centering}
\end{figure}

Now, we want to find a maximum matching in $G$ that corresponds to a proper mapping of letters between the strings, that is, such that every two consecutive mapped duos (consisting of three consecutive letters)
are mapped to two consecutive duos (in the same order). It is not necessary that all duos are mapped.
Formally, a matching $M$ is called consecutive if every two neighbouring nodes are either matched to two neighbouring nodes (preserving the order) or at least one of them is unmatched:
\[\bigforall_{i,j,j'\in\{1..n\}}\big(\pair{a_i}{b_j}\in M \wedge \pair{a_{i+1}}{b_{j'}} \in M\big)\Rightarrow \big(j'=j+1\big)\]
and a symmetric condition holds for the other side of the graph.
Even though the graph $G$ obtained as described above from an instance of MPSM
has some additional structure, we focus only on the more general problem
where the given bipartite graph $G=(A\dot{\cup}B,E)$ is arbitrary
and we are looking for a consecutive matching of maximum cardinality.
This was called the Maximum Consecutive Bipartite Matching (MCBM) by Boria et al.~\cite{Boria2014}.

\paragraph{Definitions.}
We say that two edges $\pair{a_i}{b_j}$ and $\pair{a_{i'}}{b_{j'}}$ are overlapping if
$|i-i'| \le 1 $ or $|j-j'|\le 1$. 
Given a consecutive matching $M$, we define
a streak to be a maximal (under inclusion) set of \emph{consecutive} edges
$e_1,e_2,\ddd,e_k$, such that for some $p,q$ we have that $e_i=\pair{a_{p+i}}{b_{q+i}}$
for all $i=1,2,\ldots,k$. See Figure~\ref{fig:confl_and_streak}.
Note that from the definition, $e_i$ overlaps with itself, $e_{i-1}$ and $e_{i+1}$ (assuming that
these edges exist).
This notion is extended to sets of edges: $S_1$ overlaps with $S_2$ if there exist
$e_1\in S_1, e_2\in S_2$ such that $e_1$ overlaps with $e_2$. We can similarly define overlaps
between an edge and a set of edges. Note that every consecutive matching $M$ can be uniquely
decomposed into a set of streaks such that no two of them are overlapping with each other.

\begin{figure}[h]
\begin{centering}
\includegraphics{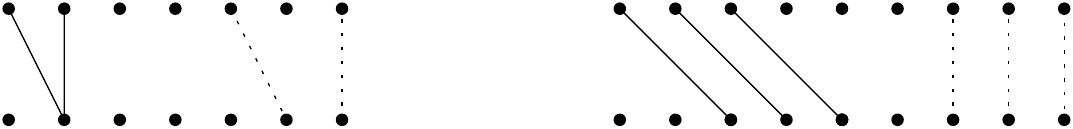}
\caption{Two pairs of overlapping edges (left) and decomposition of a consecutive matching into streaks (right).}
\label{fig:confl_and_streak}
\end{centering}
\end{figure}

\section{Greedy Algorithm}\label{greedy_alg}

Consider a simple greedy procedure, that in every step takes the
longest possible streak from $G$ and, if the streak consists of at least $k$ edges, adds
it to the solution. See Algorithm~\ref{alg:greedy}.

\begin{algorithm}[H]
\begin{algorithmic}[1]
  \Function{Greedy}{$k$}
  \State $\ALG := \emptyset$
  \While{\texttt{true}}
    \State $s:= $ the largest streak in $G$
    \If{$|s|<k$}
    \State \textbf{break}
    \EndIf
    \State remove $s$ and all edges overlapping with $s$ from $G$
    \State $\ALG := \ALG \cup s$
  \EndWhile
  \State \Return $\ALG$
  \EndFunction  
\end{algorithmic}
\caption{Choosing the largest possible streak greedily.}
\label{alg:greedy}
\end{algorithm}

To analyse quality of the returned solution, we fix an optimal solution $\OPT$ and would like to compare $|\ALG|$ with $|\OPT|$.
Let $s_i$ be the streak that was removed in the $i$-th step of the algorithm
and $o_i$ be the set of edges from $\OPT$ that are overlapping with $s_i$, but
were not overlapping with $s_1,s_2,\ldots,s_{i-1}$.  In other words,
$o_i$ consists of those edges from $\OPT$ that
after $i-1$ steps of the algorithm still could have been added to the
solution, but are no longer available after the $i$-th step.
Note that $o_i$ contains all the edges of $\OPT\cap s_i$, because every edge overlaps with itself.
Observe that $|o_i| \leq 2|s_i| + 4$ as there can be at most $|s_i|+2$ edges from $o_i$
overlapping with $s_i$ at each side of $G$.
Moreover, even a stronger property holds:

\begin{lemma}
\label{Lemma_oi}
$|o_i| \leq 2|s_i| + 2$.
\end{lemma}

\begin{proof}
Suppose that the endpoints of $s_i$ at one side (say $A$) of the bipartite graph form
a sequence of nodes $a_j,a_{j+1},\ldots,a_{j+|s_i|-1}$.
Define $\mathcal{E}~=~\{a_{j-1},a_j,\ldots,a_{j+|s_i|-1},a_{j+|s_i|}\}$
(assuming that $a_{j-1}$ and $a_{j+|s_{i}|}$ exist).
We will show that at most $|s_i| + 1$ edges from $o_i$ can end in $\mathcal{E}$.
Then, applying the same reasoning to the other side of the graph will finish the proof.
If $|\mathcal{E}|<|s_i|+2$ then the claim holds.
Otherwise, if $|\mathcal{E}|=|s_i|+2$, there are three cases to consider:
\begin{enumerate}
  \item There are two or more streaks from $o_i$ ending in $\mathcal{E}$. Then they cannot end in all nodes from~$\mathcal{E}$, because
    at least two of them would be overlapping with each other. Thus there is at least one node from~$\mathcal{E}$ that is not an endpoint of edge from $o_i$,
    so there are at most $|s_i| + 1$ of them.
  \item There is one streak from $o_i$ ending in $\mathcal{E}$. Then the streak cannot be larger than $|s_i|$,
    because then the greedy algorithm would have taken the larger streak
    (recall that $o_i$ consists of edges that could have been added to the solution in the $i$-th step).
    Thus there are at most $|s_i|$ edges of $o_i$ ending in $\mathcal{E}$.
  \item There is no streak from $o_i$ ending in $\mathcal{E}$. Then the statement holds trivially. \qedhere
\end{enumerate}
\end{proof}

We still need to specify the algorithm for smaller streaks (consisting of less than $k$ edges),
but before doing so, in the next section we bound the quality of the
solution found by the greedy algorithm.

Let $m$ be the number of steps performed by the greedy algorithm. The algorithm
returns $\ALG=\bigcup_{i=1}^m s_i$ which should be compared with the set
of edges of $\OPT$ that can no longer be taken due to the decisions made by the
greedy algorithm, that is, $\bigcup_{i=1}^m o_i\subseteq OPT$. Using Lemma~\ref{Lemma_oi}
we can compute the desired ratio as follows:

\[
\frac{|\bigcup_{i=1}^{m} o_i|}{|\bigcup_{i=1}^{m} s_i|} = \frac{\sum_{i=1}^{m} |o_i|}{\sum_{i=1}^{m} |s_i|}
\leq \frac{\sum_{i=1}^{m}{(2|s_i|+2)}}{\sum_{i=1}^{m}{|s_i}|} = 2 + \frac{m\cdot 2}{\sum_{i=1}^{m} |s_i|}
\leq 2 + \frac{m\cdot 2}{m\cdot k} = 2+\frac{2}{k}
\]

\noindent where the last inequality holds because all taken streaks consist of at least $k$ edges.

To conclude, the solution $\ALG$ found by the greedy algorithm is at most
$2+\frac{2}{k}$ times smaller than the set of edges from $\OPT$ that is
overlapping with $\ALG$. Informally, on average we discard only a few edges of
$\OPT$ for every edge from $ALG$.
After running the algorithm for $k=1$, there will be no
edges left and thus we have a simple $4$-approximation algorithm. To obtain a better
approximation ratio, we will increase $k$ and focus on the subgraph $G'$ of $G$
consisting of all edges that are not overlapping with any streak $s_{i}$ already taken
by the algorithm (and hence still available). The crucial insight is that
we can analyse the performance of the greedy algorithm on $G\setminus G'$
and the performance of the algorithm for small $k$ on $G'$ separately.
We know that the approximation ratio of the greedy algorithm on $G\setminus G'$
is $2+\frac{2}{k}$ and size of the optimal solution for $G'$ is at least $|\OPT\setminus\bigcup_{i=1}^{m} o_i|$.
Then, due to the definition of $G'$, any solution found
for $G'$ can be combined with $\ALG$ to obtain a solution for the original
instance, so the final approximation ratio is the maximum of $2+\frac{2}{k}$
and the ratio of the algorithm used for $G'$.

\section{Algorithm for Small \texorpdfstring{$k$}{k}}

As stated above, applying the greedy algorithm with $k=1$ immediately implies a $4$-approximation algorithm.
For larger values of $k$ we need another phase to find a solution for the remaining part of the graph.
For $k=2$, we present a simple algorithm based on maximum bipartite matching
(not consecutive) that can be used to obtain a 3-approximation.
For larger values of $k$, we first consider $k=3$ and design a quadratic-time algorithm
based on the local search technique. Then, we move to a general $k$ and develop
a more involved polynomial-time algorithm that achieves $(2+\eps)$-approximation.

\subsection{3-approximation Based on Maximum Matching for \texorpdfstring{$k=2$}{k=2}}

After running \textsc{Greedy}(2) there are no streaks of size 2. Recall that
$G'=(A \dot{\cup} B, E')$ is the subgraph of the original graph $G$ consisting of all edges
that are not overlapping with the edges already taken. Consider the following algorithm:

\begin{enumerate}
  \item Create a bipartite graph $H=(A' \dot{\cup} B', F)$ where:
    \begin{itemize}
    \item $A'=\{a_{(1,2)},a_{(3,4)},\ldots,a_{(n-1,n)}\}$ and similarly for $B'$. In other words, nodes of $A'$
    correspond to merged pairs of neighbouring nodes of $A$ (if $n$ is odd, the last node of $A'$ corresponds
    to a single node of $A$).
    \item $F=\Big\{ \{a_{(2i-1,2i)},b_{(2j-1,2j)}\}: \{a_{2i-1},a_{2i}\}\times\{b_{2j-1},b_{2j}\}\cap E' \ne \emptyset \Big\}$.
    In other words, there is an edge between two merged pairs of nodes
    if there was an edge between a node from the first pair and a node from the second pair.
    \end{itemize}
  \item Find a maximum matching $M'$ in $H$.
  \item For every edge of $M'$, choose an edge of $G'$ connecting nodes from the corresponding pairs
  (if there are multiple possibilities, choose any of them). Let $M$ be the set of chosen edges.
  \item Let $\ALG\gets\emptyset$. Process all edges of $M$ in arbitrary order. For an edge $(a_i,b_j)\in M$:
    \begin{itemize}
    \item remove from $M$ all edges ending in nodes $a_{i-1},a_{i+1},b_{j-1}$ and $b_{j+1}$,
    \item add $(a_i,b_j)$ to $\ALG$.
    \end{itemize}
  \item Return $\ALG$.
\end{enumerate}

Consider the optimal solution $\OPT$. As $G'$ contains no streaks consisting of 2 or more edges,
the endpoints of any two of its edges cannot be neighbouring. Therefore, $\OPT$ can be translated
into a matching in $H$ with the same cardinality, so $|\OPT| \leq |M'|$.

We claim that after including an edge $(a_i,b_j)\in M$ in $\ALG$ at most 2 other edges are removed from
$M$. Assume otherwise, that is, there are 3 such edges. Without loss of generality, one of them
ends in $a_{i-1}$ and one in $a_{i+1}$. Depending on the parity of $i$, edge $(a_{i},b_{j})$
and the edge ending in either $a_{i-1}$ or $a_{i+1}$ correspond
in $H$ to edges ending in the same node of $A'$. This is a contradiction, because all edges
in $M'$ have distinct endpoints. Because initially $|M'|=|M|$, we conclude that $|\ALG| \geq |M'| /3$.

Combining the inequalities gives us $3\cdot |ALG|\ge |M'| \ge |\OPT|$, so the above algorithm
is a 3-approximation for graphs with no streaks of size of at least $2$. Combining it with \textsc{Greedy}(2),
that guarantees approximation ratio of $2+\frac{2}{k}=2+\frac{2}{2}=3$, gives us a 3-approximation
algorithm for the whole problem.

\subsection{2.67-approximation for \texorpdfstring{$k=3$}{k=3}}

For $k=3$ we use procedure \textsc{LocalImprovements} based on the local search technique.
See Algorithm~\ref{alg:local_improvements}.
Essentially the same method was used to obtain the $3.5$-approximation~\cite{DuaItalians}.
The algorithm consists of a number of steps
in which it tries to either add a single edge or remove one edge so that two other edges
can be added.
However, the crucial difference is that
in our case there are no streaks of size greater than $2$ in $G'$. This allows for a better
bound on the approximation ratio.

\begin{algorithm}[h]
\begin{algorithmic}[1]
  \Function{LocalImprovements}{}
  \State $\ALG := \emptyset$
  \While{\texttt{true}}
    \If{$\exists e \notin \ALG$ s.t. $\ALG \cup \{e\}$ is a valid solution}
    \State $\ALG:=\ALG \cup \{e\}$
    \EndIf
    \If{$\exists e_1,e_2 \notin ALG,e'\in \ALG$ s.t. $\ALG \setminus \{e'\} \cup \{e_1,e_2\}$ is a valid solution}
    \State $\ALG:=\ALG \setminus \{e'\} \cup \{e_1,e_2\}$
    \EndIf
    \If{$|\ALG|$ did not increase}
    \State \textbf{break}
    \EndIf
  \EndWhile
  \State \Return $\ALG$
  \EndFunction
\end{algorithmic}
\caption{Local improvements in $\Oh(m^2n^2)$ time.}
\label{alg:local_improvements}
\end{algorithm}

Fix an optimal solution $\OPT$. We want to bound the total number $C$ of overlaps between
the edges from $\ALG$ and $\OPT$.
First, observe that an edge from $\ALG$ can overlap with at most $4$ edges from $\OPT$,
because there are no streaks of size $3$ in the graph. Thus:
\begin{equation}\label{conflicts1}
\begin{gathered}
4\cdot |\ALG| \ge C.
\end{gathered}
\end{equation}
Second, let $k_1$ be the number of edges from $\OPT$ that overlap with exactly one edge from $\ALG$.
Then all other edges from $\OPT$ overlap with at least two edges from $\ALG$
(because otherwise the algorithm would have taken an edge not overlapping with any edge already taken), so:
\begin{equation}\label{conflicts2}
C \ge k_1 + 2\cdot(|\OPT|-k_1) = 2\cdot |\OPT| - k_1.
\end{equation}

\begin{lemma}
\label{le:for_local_improvements}
$k_1 \le |\ALG|.$
\end{lemma}

\begin{proof}
Suppose that $k_{1} > |\ALG|$.
Then there are two edges $e_1,e_2\in \OPT$ that overlap with only one and the
very same edge $e_{del}\in\ALG$. But then the algorithm would be able to increase size of the solution
by removing $e_{del}$ and adding $e_1$ and $e_2$, so we obtain a contradiction.
\end{proof}

Applying Lemma~\ref{le:for_local_improvements} to \eqref{conflicts2} and combining with \eqref{conflicts1}
we get $4\cdot |\ALG| \ge C \ge 2\cdot |\OPT| - |\ALG|$
and thus $2.5 \cdot |\ALG| \ge |\OPT|$.
Recall that the approximation ratio of the first greedy part of the algorithm is $2+\frac{2}{3}<2.67$,
so the overall ratio of the combined algorithm is also $2.67$.
The algorithm clearly runs in polynomial time as in every iteration of the main loop
the size of $\ALG$ increases by one and is bounded by $n$. In \cite{DuaItalians} the running
time was further optimised to $\Oh(n^{4})$, but in the remaining part of this section we will describe how
to decrease the time to $\Oh(n^{2})$. We will also show how to implement the greedy algorithm
in the same $\Oh(n^{2})$ complexity, thus obtaining an $2.67$-approximation algorithm
in $\Oh(n^{2})$ time.
\newpage
\paragraph{Greedy part in $\Oh(n^2)$ time.}
We show how to implement \textsc{Greedy}($k$) in $\Oh(n^2)$ time.
Recall that in every iteration the algorithm chooses the longest streak in the remaining part of the graph,
includes it in the solution, and removes all edges that overlap with it from the graph.
The procedure terminates if the streak contains less than $k$ edges.

We start with creating a list $L$ of edges $\pair{x}{y}$ sorted lexicographically first by 
$x$ and then by $y$.
This can be done in $\Oh(n^2)$ time using bucket sort and while sorting we can also retrieve
for every edge the edge that would be its predecessor in a streak.
Then we iterate over the edges in $L$ and split them into streaks.
The edges of every streak are stored in a doubly linked list and every edge stores a pointer to its streak.
We also keep streaks grouped by size, that is, $D_s$ contains all streaks of size $s$.
To allow insertions and deletions in $\Oh(1)$ time, $D_s$ is internally also implemented as a doubly linked list,
but in order not to confuse it with the lists storing edges inside a streak, later on we will refer to lists $D_s$ as groups.


Having split all edges into streaks and having grouped streaks by their sizes, we iterate over the groups $D_n,D_{n-1},\ldots,D_k$ and retrieve a streak $s$ from the non-empty group with the largest index.
We add $s$ to the solution and remove all edges overlapping with $s$ from the graph.
Every removed edge either decreases the size of its streak by one or splits it into two smaller
streaks. In both cases, the smaller streak(s) is moved between the appropriate groups.
Removing an edge takes constant time and every edge is removed from the graph at most once.
Similarly, moving or splitting a streak due to a removed edge takes constant time
as the size of the smaller streak can be computed in constant time
by looking at its first and last edge.
Thus, the overall time of the procedure is $\Oh(n^2)$.

\paragraph{Remark.}
Recall that we have generalised the MPSM problem and we are now working with an arbitrary bipartite graph $G$.
However, if $G$ was constructed from an instance of MPSM, then finding the longest streak available corresponds
to finding the longest string that occurs in both $X$ and $Y$ without overlapping with any of the previously
chosen substrings.
Goldstein and Lewenstein~\cite{GoldsteinLewenstein2014} showed how to implement such a procedure in $\Oh(n)$ total time.

\paragraph{Local improvements in $\Oh(n^2)$ time.}

Recall that to analyse the approximation ratio (in Lemma~\ref{le:for_local_improvements}), we only 
need that after termination of the algorithm there are no three edges $e_1,e_2\notin \ALG, e_{del}\in \ALG$
such that $\ALG\setminus\{e_{del}\}\cup\{e_1,e_2\}$ is a valid solution.
At a high level, \textsc{FastLocalImprovements} keeps track of edges that can potentially
increase size of the solution in a queue $Q$.
As long as $Q$ is not empty, we retrieve a candidate edge $e$ from $Q$. First, we verify
that $e\notin \ALG$ and $e$ overlaps with at most one edge from $\ALG$. If $e$ can be
added to $\ALG$, we do so and continue after adding to $Q$ all edges overlapping with $e$.
Otherwise, we check if some other edge $e'$ can be added
while removing another edge $e_{del}$ and at the same time using procedure \textsc{TryAddingPairWith}($e$),
and if so, we add to $Q$ all edges overlapping with one of the modified edges ($e,e'$ and $e_{del}$).
See Algorithm~\ref{alg:fast_local_improvements} and Algorithm~\ref{alg:try_adding_pair_with}.

\begin{algorithm}[H]
\begin{algorithmic}[1]
  \Function{FastLocalImprovements}{}
  \State $Q.\textsc{enqueue}(E)$
  \While {$Q$ is not empty}
    \State $e:=Q.\textsc{dequeue}()$
    \If{$e\in \ALG$ or $e$ overlaps with more than one edge from $\ALG$}
      \State \textbf{continue}
    \EndIf
    \If{$\ALG \cup \{e\}$ is a valid solution}
      \State $\ALG:= \ALG \cup \{e\}$
      \State $Q.\textsc{enqueue}\big(\confl(e)\big)$
      \State \textbf{continue}
    \EndIf
    \State \textsc{TryAddingPairWith}($e$)
  \EndWhile  
  \EndFunction
\end{algorithmic}
\caption{Local improvements in $\Oh(n^2)$ time.}
\label{alg:fast_local_improvements}
\end{algorithm}

The algorithm uses the following data structures and functions:
\begin{itemize}
  \item For every node $v\in G'$, we keep a list of all edges from $E$ ending in $v$ and separately edges of $\ALG$ ending in $v$.
  \item $\close(e)$ is the set of nodes of $G'$ at a distance of at most $1$ from the endpoints of edge $e$. In other words, $\close(e)$ is the set of up to 6 nodes where edges overlapping with $e$ can end.
  \item $\confl(e)$ is the set of edges overlapping with edge $e$. It is computed on the fly, by iterating through edges ending in $v\in \close(e)$.
  \item Queue $Q$ of candidate edges. For every edge in $E$ we remember if it is currently in $Q$ in order not to store any duplicates
  and keep the space usage $\Oh(m)$.
  \item For every node $v\in G'$ we keep a list $L_v$ of edges from $E\setminus \ALG$ that overlap with exactly one edge from $\ALG$ and end in $v$. 
  To keep these lists updated, every time an edge $e=\pair{x}{y}$ is enqueued or added or removed from $\ALG$, we count the edges from $\ALG$ it overlaps with.
  If there is only one of them, we make sure that $e$ is in $L_x$ and $L_y$, otherwise we remove $e$ from $L_x$ and $L_y$.
\end{itemize}

\begin{algorithm}[h]
\begin{algorithmic}[1]
  \Function{TryAddingPairWith}{$e$}
  \State $e_{del} :=$ the only edge from $\ALG$ overlapping with $e$
  \ForEach{$e'$ that can be a neighbour of $e$ in a streak} \label{li:loop_1} \Comment{$\Oh(1)$}
      \If{$\ALG \setminus \{e_{del}\} \cup \{e,e'\}$ is a valid solution}
	\State $\ALG:=\ALG \setminus \{e_{del}\} \cup \{e,e'\}$
	\State $Q.\textsc{enqueue}\big(\confl(e)\cup \confl(e')\cup \confl(e_{del})\big)$ \label{li:enqueue1}
	\State \Return
      \EndIf
    \EndFor
  \ForEach{node $v\in \close(e_{del})\setminus \close(e)$} \label{li:loop_2} \Comment{$\Oh(1)$}
      \ForEach{edge $e'\in L_v$} \Comment{see Lemma~\ref{le:loop_is_short}} \label{li:tricky_line}
	\If{$\ALG \setminus \{e_{del}\} \cup \{e,e'\}$ is a valid solution}
	  \State $\ALG:=\ALG \setminus \{e_{del}\} \cup \{e,e'\}$
	  \State $Q.\textsc{enqueue}\big(\confl(e)\cup \confl(e')\cup \confl(e_{del})\big)$ \label{li:enqueue2}
	  \State \Return
	\EndIf
      \EndFor
    \EndFor
  \EndFunction
\end{algorithmic}
\caption{Adding a pair with edge $e$.}
\label{alg:try_adding_pair_with}
\end{algorithm}

Clearly, after termination of the algorithm there is no triple of edges $e_1,e_2$ and $e_{del}$ that can be used to increase the solution, because
every time an edge is added to or removed from the solution, all of its overlapping edges are enqueued.
It remains to prove that Algorithm~\ref{alg:fast_local_improvements} indeed runs in $\Oh(n^2)$ time.
First, observe that $|\close(e)|\le~6$, so from the definition of overlapping
edges $|\confl(e)| \le |\close(e)| \cdot n\in \Oh(n)$, as there are at most $n$ edges ending in a node.
So, every time the algorithm enqueues a set of edges, there are at most $\Oh(n)$ of them.
As this happens only after increasing the size of $\ALG$, which
can happen at most $n$ times, in total there are $\Oh(n^2)$ enqueued edges.
So it suffices to prove that every time an edge $e$ is dequeued, it takes
$\Oh(1)$ time to check if it can be used to increase the solution. Here we disregard
the time for enqueuing edges due to increasing the size of $\ALG$, as 
it adds up to $\Oh(n^{2})$ as mentioned before. Note that both counting the edges overlapping 
with $e$ and finding the unique edge from $\ALG$ overlapping with $e$ takes $\Oh(1)$ time,
as we just need to check edges from $\ALG$ ending in $\close(e)$.
Similarly, as $\ALG$ is always a valid solution, each validity check takes $\Oh(1)$
time, as we always try to modify a constant number of edges.
By the same argument, loops in lines~\ref{li:loop_1} and~\ref{li:loop_2} take constant number
of iterations, and also:

\begin{lemma}
\label{le:loop_is_short}
There are $\Oh(1)$ iterations of the loop in line~\ref{li:tricky_line} of \textsc{TryAddingPairWith}($e$).
\end{lemma}

\begin{proof}
Consider an edge $e'\in L_v$ such that $\ALG':=\ALG \setminus \{e_{del}\} \cup \{e,e'\}$ is not a valid solution.
From the definition of $L_{v}$, $e'$ overlaps only with $e_{del}\in \ALG$, so both 
$\ALG \setminus \{e_{del}\} \cup \{e\}$ and $\ALG \setminus \{e_{del}\} \cup \{e'\}$ are valid solutions.
Thus, the only reason for $\ALG'$ not being valid is that $e'$ overlaps with $e$.
But $v$ is at a distance of $2$ or more from the endpoint of $e$, so $e$ and $e'$ can be overlapping only at the other side of the graph.
There are at most 3 possible endpoints of such $e'$ at the other side, see Figure~\ref{fig:loop_lemma}. Consequently, after checking 4 edges from $L_v$ we will surely find one that can be used to increase $|\ALG|$.
\end{proof}

\begin{figure}[h]
\begin{centering}
\includegraphics{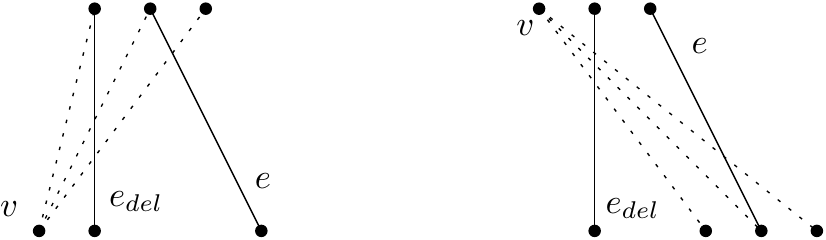}
\caption{Dotted lines show the only $3$ possible edges $e' \in L_v$ that overlap with $e$. 
Among any $4$ edges in $L_v$, at least one can be used to increase $|\ALG|$ and break the loop.}
\label{fig:loop_lemma}
\end{centering}
\end{figure}

To conclude, \textsc{Greedy}$(3)$ with \textsc{FastLocalImprovements} yield
$2.67$-approximation in $\Oh(n^{2})$ time.

\section{\texorpdfstring{$(2+\eps)$}{(2+eps)}-approximation}

Given $\eps>0$ we would like to create a polynomial time $(2+\eps)$-approximation algorithm.
We set $k = \lceil \frac{2}{\eps} \rceil$ and run \textsc{Greedy}$(k)$ to remove all streaks
of size at least $k$ from the graph $G$. From now on we are going to focus on the subgraph $G'$ remaining after the
first greedy phase and let $\OPT$ denote the optimal solution in $G'$.

Let $t=\lceil \frac{4}{\eps}\rceil +1$ and $\ALG$ be the solution found by
\textsc{BoundedSizeImprovements}($t$), see Algorithm~\ref{alg:bounded_size_improvement}.
Similarly to the case $k=3$, the algorithm tries to improve the current solution using
local optimisations, however now the number of edges that we try to add or remove in every
step is bounded by $t$ (that depends on $\eps$).
We want to prove that $(2+\eps)\cdot |\ALG| \geq |\OPT|$. To this end, we assign $(2+\eps)$
units of credit to every edge of $\ALG$. Then the goal is to distribute the
credits from the edges of $\ALG$ to the edges of $\OPT$, so that every edge of $\OPT$ 
receives at least one credit. Alternatively, we can think of transferring credits to the streaks
from $\OPT$ in such a way that a streak consisting of $s$ edges receives at least $s$ credits.
This will clearly demonstrate the required inequality.

\begin{algorithm}[h]
\begin{algorithmic}[1]
  \Function{BoundedSizeImprovements}{$t$}
  \State $\ALG := \emptyset$
  \While{\texttt{true}}
  \ForEach{ $E_{\text{remove}},E_{\text{add}}\subseteq E \text{ such that } |E_{\text{remove}}|<|E_{\text{add}}|\leq t$}
      \State $\ALG' := \ALG\setminus E_{\text{remove}}\cup E_{\text{add}}$
      \If{$\ALG'$ is a valid solution}
	\State $\ALG:=\ALG'$
	\State \textbf{break}
      \EndIf
    \EndFor
    \If{$|\ALG|$ did not increase}
      \State \textbf{break}
    \EndIf
  \EndWhile
  \State \Return $\ALG$
  \EndFunction
\end{algorithmic}
\caption{Improvements of bounded size.}
\label{alg:bounded_size_improvement}
\end{algorithm}

\paragraph{Credit distribution scheme.}
Every edge from $\ALG$ distributes $(1+\frac{\eps}{2})$ credits from each of its two endpoints independently.
Consider an endpoint $v_i$ of an edge from $\ALG$. Let $\ddd, v_{i-1},v_i,v_{i+1},\ddd$ be all nodes at the
corresponding side of the graph $G$.
If there is an edge $e\in\OPT$ ending in $v_i$, then 
$e$ receives 1 credit. Now consider the case when no edge of $\OPT$ ends in $v_i$. If exactly one edge from $\OPT$ ends in $v_{i+1}$ or $v_{i-1}$ then the credit is transferred to that edge. If there are
no edges ending there then the credit is not transferred at all. Finally, if there is an edge $e\in\OPT$ ending at $v_{i-1}$ and another edge
$e'\in\OPT$ ending at $v_{i+1}$, then for the time being neither $e$ nor $e'$ receives the credit. 
In such a situation we say that the node $v_i$ is between the streak containing $e$ and the streak containing $e'$,
call the credit \emph{uncertain} and defer deciding whether it should be transferred to $e$ or $e'$.
Observe that the only case when an edge $e\in\ALG$ overlapping with a streak $s$ does not transfer the credit to $s$ is when
the endpoint of $e$ is between two streaks $s$ and $s'$, see Figure~\ref{fig:scheme}.
Note that two credits can be transferred from $e$ to $s$ if both endpoints of $e$ transfer its credits to~$s$.
The remaining $\frac{\eps}{2}$ credits are not transferred to any specific edge yet. We will aggregate and redistribute them using a more global argument, but first we need some definitions.

\begin{figure}[h]
\begin{centering}
\includegraphics{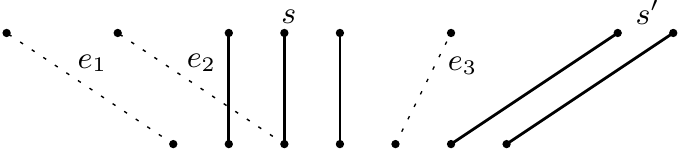}
\caption{Dotted lines denote edges from $\ALG$.
According to the scheme, $e_1$ and $e_2$ transfer a credit to an edge from $s$ and $e_3$ transfers its credit neither to $s$ nor to $s'$, as its endpoint is between $s$ and $s'$.}
\label{fig:scheme}
\end{centering}
\end{figure}

\paragraph{Gaps and balance.}
Define the balance of a streak $s$ from $\OPT$ as the number of credits obtained
in the described scheme (ignoring the uncertain credits) minus the number of edges in $s$.
A gap is an edge of $\OPT$ that has not received any credits yet and
$\gaps(s)$ is the number of gaps in $s$.
Observe that the balance of a streak $s$ is at least $-\gaps(s)$.
After running the greedy algorithm and \textsc{BoundedSizeImprovements}$(t)$, even a stronger property holds:

\begin{lemma}
The balance of every streak is at least $-2$.
\end{lemma}

\begin{proof}
Consider a streak $s$. If there are less than $2$ gaps in $s$ then the claim holds.
Otherwise, let $g_1$ and $g_{2}$ be the first and the last gap in $s$, so that we can write
$s=Ag_1Mg_2B$, see Figure~\ref{fig:split}. Note that the balance of both $A$ and $B$ is non-negative, as from the definition
there are no gaps inside, so every edge there receives at least one credit.
However, there might be multiple gaps in $M$. Suppose that the balance of $M$ is negative.
But the size of $M$ is smaller than $k<t$, so \textsc{BoundedSizeImprovements}$(t)$ would have
replaced a subset of edges from $\ALG$ with $M$ to increase size of the solution. Therefore,
the balance of $M$ is nonnegative. Finally, observe that the balance of $s$ is equal to the sum of
balances of $A,M$ and $B$ minus $2$ (for the gaps $g_1$ and $g_2$), so it is at least $-2$ in total.
\end{proof}

The following corollary that follows from the above proof will be useful later:

\begin{corollary}
\label{split_streak}
Every streak $s$ with balance $-2$ can be represented as $s=Ag_1Mg_2B$ where $g_1$ and $g_2$
are the first and last gap of $s$, respectively. The balance of $Ag_{1}$ and $g_{2}B$ is $-1$
while the balance of $M$ is 0.
\end{corollary}

\begin{figure}[h]
\begin{centering}
\includegraphics{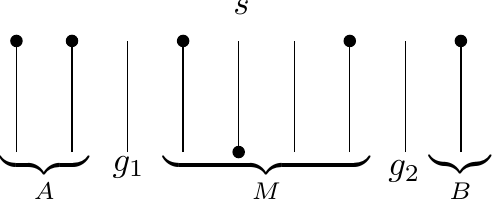}
\caption{Black dots denote endpoints of edges from $\ALG$, $g_1$ and $g_2$ are the first and the last gap, respectively.
Balance of the streak $s$ is $-3$.}
\label{fig:split}
\end{centering}
\end{figure}	

\paragraph{Analysis of the scheme.}
We construct an auxiliary multi-graph $H$, where the vertices are streaks of $\OPT$
with balance at least $-1$. Streaks with balance $-2$ are split into two smaller streaks (called substreaks)
with balance $-1$ as explained in Corollary~\ref{split_streak}.
We create an edge between two streaks in $H$ when they both overlap with an endpoint of an edge from $\ALG$.
In other words, when an edge $e$ from $\ALG$ has an endpoint $x$ overlapping with two streaks of $\OPT$, then there is an edge in $H$ between the vertices corresponding to these streaks, see Figure~\ref{fig:edge}.
Observe that then there is no edge of $\OPT$ ending in $x$ and there can be at most two edges
between any pair of streaks in $H$.

\begin{figure}[h]
\begin{centering}
\includegraphics{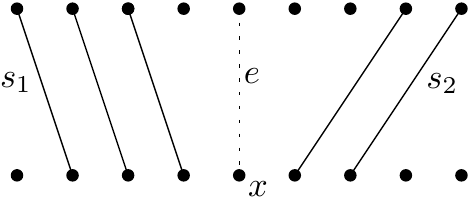}
\caption{If there is an endpoint $x$ of edge $e\in \ALG$ that is between two streaks $s_1,s_2$ of $\OPT$ then we add an edge between $s_1$ and $s_2$ in $H$.}
\label{fig:edge}
\end{centering}
\end{figure}

Now we will show that for every connected component of $H$ there are enough credits to distribute
at least one credit to every edge from $\OPT$ in the component.
The intuition behind considering the connected components of $H$ is that we have deferred distribution of the uncertain credits,
and now a connected component is a set of streaks that need to decide together how to spend those uncertain credits.
At a high level, for every connected component $\CC$ of $H$ there will be two cases to consider.
First, if the balance of $C$ is non-negative, then we are done.
Otherwise, we will show that the balance of $\CC$ is equal to $-1$.
We also know that the component is so big that \textsc{BoundedSizeImprovements}
was not able to increase the solution.
From this we will conclude that, by gathering the remaining $\frac{\eps}{2}$ credits 
together, it is possible to cover the deficit.

Consider one connected component $\CC$ on $w$ vertices.
We want to prove that there are at least $w$ credits transferred to all edges of $\CC$ in total.
From the construction we have that every vertex of $\CC$ has balance of at least $-1$.
Moreover, as the component is connected, there are at least $w-1$ edges, each adding one
uncertain credit. Thus, the total balance of the whole component (including the
uncertain credits) is at least $-1$.
Observe that the only case when the total balance of the component is $-1$ is a tree
(with exactly $w-1$ edges) where every node has balance of $-1$. In all other cases
the balance is non-negative already.

We denote by $K_\CC$ the set of edges of $\OPT$ from all vertices of $\CC$ (recall that
they correspond to original streaks with balance -1 and substreaks). We also define an auxiliary set $M_\CC$ that consists of
the middle parts $M$ of the original streaks.
More precisely, for every streak $s$ of balance $-2$, if it was a part of $\CC$ (due to the substreak
$Ag_1$ or $g_2B$, where $s=Ag_1Mg_2B$), we add to $M_\CC$ all edges from $M$.
From Corollary~\ref{split_streak}, the balance of every such $M$ is $0$.
Now consider the following set of edges $X_\CC=K_\CC \cup M_\CC$.
There are two cases to consider depending on how many credits have been transferred to $X_\CC$:

\begin{enumerate}
\item If there are at least $c \geq \frac{4}{\eps}$ credits transferred to the edges of $X_\CC$ (each credit from an endpoint of an edge from $\ALG$),
then we can use half of the remaining $\frac{\eps}{2}$ credit of each endpoint and transfer it
to the component. Note that for each credit from those $c$ already assigned to $X_\CC$ there is one endpoint still having 
additional $\frac{\eps}{4}$ credit that can be spent on $X_\CC$. We can use only half of the remaining
$\frac{\eps}{2}$ credit because
some edges (from the middle parts of original streaks) can belong to both $X_\CC$ and $X_{\CC'}$ for two different components $\CC$ and $\CC'$, see Figure~\ref{fig:split2}, and they might need to transfer additional credit to both of them.
Thus, for each of the $c$ credits we transfer additional $\frac{\eps}{4}$ credit, so in total we transfer at least one full credit,
which is enough to cover the deficit of the component.

\begin{figure}[t]
\begin{centering}
\includegraphics[width=\textwidth]{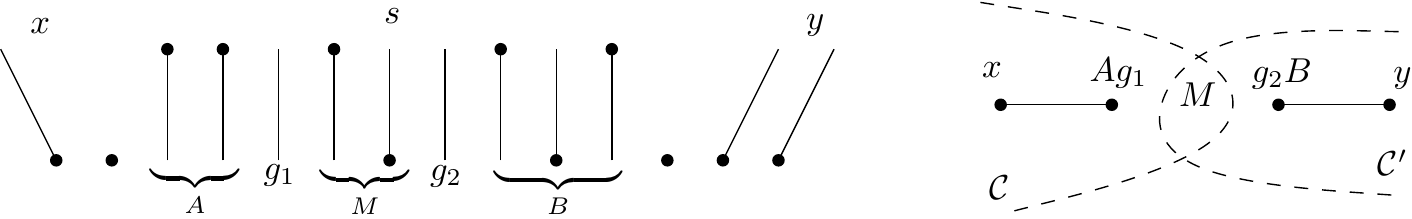}
\caption{As there is an \textit{uncertain} credit between streaks $x$ and $Ag_1$, there will be an edge between them in $H$, so they will be in a connected component $\CC$ of $H$.
Similarly for $g_2B$ and $y$ in~$\CC'$.
Observe that the middle part $M$ of the split streak $s$ is accounted for in both $M_{\CC}$ and $M_{\CC'}$.}
\label{fig:split2}
\end{centering}
\end{figure}

\item In the second case, the edges from $X_\CC$ received less than $\frac{4}{\eps}$ credits, so there are less than $\frac{4}{\eps}+1$ edges from $\OPT$
(recall that the overall balance of the component is $-1$).
Note that if we add all edges from $X_\CC$ and remove all edges from $\ALG$ that have transferred credits to the edges from $X_\CC$,
 the size of the solution will increase as earlier the overall balance was negative.
 The solution will still be valid, because we have removed all edges from $\ALG$ overlapping with the edges of $X_\CC$. Also for the split streaks,
we took edges up to (but not including) a gap which from the definition does not share an endpoint with an edge from $\ALG$.
Furthermore, as the size of $X_\CC$ is at most $\frac{4}{\eps} + 1 \leq t$, it 
would have been considered as the set $E_{add}$ of edges to be checked by our algorithm.
Thus, this case is impossible, as we would have been able to improve the current solution.
\end{enumerate}

To conclude, every connected component containing $w$ edges receives at least $w$ credits,
so $(2+\eps)\cdot |\ALG| \ge |\OPT|$. As the approximation ratio of the first greedy part
is also $(2+\eps)$, as explained before the overall algorithm is an $(2+\eps)$-approximation for MPSM.
It remains to analyse its time complexity.
Let $m$ denote the number of edges of $G'$. There are at most $n$ steps of the
algorithm, as in each of them size of the solution increases by at least one and is bounded by $n$.
There are $\binom{m}{t} \in \Oh(m^t)$ candidates for $E_{add}$ and $E_{remove}$ and we can check in $\Oh(m)$ time if a given solution is valid. In total, substituting $t=\lceil\frac4\eps\rceil +1$ the total time complexity is
$\Oh(m^{2t+1})= \Oh(n^{4t+2})=\Oh(n^{\frac{16}{\eps}+6})=n^{\Oh(1/\eps)}$.

\begin{theorem}
Combining the greedy algorithm with local improvements yields a $(2+\eps)$-approximation for 
MCBM in $n^{\Oh(1/\eps)}$ time, for any $\eps>0$.
\end{theorem}

\begin{corollary}
There exists a $(2+\eps)$-approximation algorithm for MPSM running in $n^{\Oh(1/\eps)}$ time,
for any $\eps>0$.
\end{corollary}

\bibliography{p18-Dudek}

\begin{thebibliography}{10}

\bibitem{Beretta}
Stefano Beretta, Mauro Castelli, and Riccardo Dondi.
\newblock Parameterized tractability of the maximum-duo preservation string
  mapping problem.
\newblock {\em Theor. Comput. Sci.}, 646:16--25, 2016.

\bibitem{DuaItalians}
Nicolas Boria, Gianpiero Cabodi, Paolo Camurati, Marco Palena, Paolo Pasini,
  and Stefano Quer.
\newblock A 7/2-approximation algorithm for the maximum duo-preservation string
  mapping problem.
\newblock In {\em {CPM}}, volume~54 of {\em LIPIcs}, pages 11:1--11:8. Schloss
  Dagstuhl - Leibniz-Zentrum fuer Informatik, 2016.

\bibitem{Boria2014}
Nicolas Boria, Adam Kurpisz, Samuli Lepp{\"{a}}nen, and Monaldo Mastrolilli.
\newblock Improved approximation for the maximum duo-preservation string
  mapping problem.
\newblock In {\em {WABI}}, volume 8701 of {\em Lecture Notes in Computer
  Science}, pages 14--25. Springer, 2014.

\bibitem{Wabi16}
Brian Brubach.
\newblock Further improvement in approximating the maximum duo-preservation
  string mapping problem.
\newblock In {\em {WABI}}, volume 9838 of {\em Lecture Notes in Computer
  Science}, pages 52--64. Springer, 2016.

\bibitem{Bulteau2013}
Laurent Bulteau, Guillaume Fertin, Christian Komusiewicz, and Irena Rusu.
\newblock A fixed-parameter algorithm for minimum common string partition with
  few duplications.
\newblock In {\em {WABI}}, volume 8126 of {\em Lecture Notes in Computer
  Science}, pages 244--258. Springer, 2013.

\bibitem{Bulteau2014}
Laurent Bulteau and Christian Komusiewicz.
\newblock Minimum common string partition parameterized by partition size is
  fixed-parameter tractable.
\newblock In {\em {SODA}}, pages 102--121. {SIAM}, 2014.

\bibitem{Chen2014}
Wenbin Chen, Zhengzhang Chen, Nagiza~F. Samatova, Lingxi Peng, Jianxiong Wang,
  and Maobin Tang.
\newblock Solving the maximum duo-preservation string mapping problem with
  linear programming.
\newblock {\em Theoretical Computer Science}, 530(Complete):1--11, 2014.

\bibitem{Chrobak2005}
Marek Chrobak, Petr Kolman, and Ji\v{r}\'{\i} Sgall.
\newblock The greedy algorithm for the minimum common string partition problem.
\newblock {\em ACM Trans. Algorithms}, 1(2):350--366, October 2005.

\bibitem{Cormode2002}
Graham Cormode and S.~Muthukrishnan.
\newblock The string edit distance matching problem with moves.
\newblock {\em {ACM} Trans. Algorithms}, 3(1):2:1--2:19, 2007.

\bibitem{Damaschke2008}
Peter Damaschke.
\newblock Minimum common string partition parameterized.
\newblock In {\em {WABI}}, volume 5251 of {\em Lecture Notes in Computer
  Science}, pages 87--98. Springer, 2008.

\bibitem{Fu2011}
Bin Fu, Haitao Jiang, Boting Yang, and Binhai Zhu.
\newblock Exponential and polynomial time algorithms for the minimum common
  string partition problem.
\newblock In {\em {COCOA}}, volume 6831 of {\em Lecture Notes in Computer
  Science}, pages 299--310. Springer, 2011.

\bibitem{GoldsteinHardness}
Avraham Goldstein, Petr Kolman, and Jie Zheng.
\newblock Minimum common string partition problem: Hardness and approximations.
\newblock In {\em Journal of Combinatorics}, pages 484--495. Springer, 2004.

\bibitem{GoldsteinLewenstein2014}
Isaac Goldstein and Moshe Lewenstein.
\newblock Quick greedy computation for minimum common string partition.
\newblock {\em Theor. Comput. Sci.}, 542:98--107, July 2014.

\bibitem{He2007}
Dan He.
\newblock A novel greedy algorithm for the minimum common string partition
  problem.
\newblock In {\em {ISBRA}}, volume 4463 of {\em Lecture Notes in Computer
  Science}, pages 441--452. Springer, 2007.

\bibitem{Jiang2012}
Haitao Jiang, Binhai Zhu, Daming Zhu, and Hong Zhu.
\newblock Minimum common string partition revisited.
\newblock {\em Journal of Combinatorial Optimization}, 23(4):519--527, 2012.

\bibitem{Kaplan}
Haim Kaplan and Nira Shafrir.
\newblock The greedy algorithm for edit distance with moves.
\newblock {\em Information Processing Letters}, 97(1):23 -- 27, 2006.

\bibitem{Shapira2002}
Dana Shapira and James~A. Storer.
\newblock Edit distance with move operations.
\newblock {\em J. Discrete Algorithms}, 5(2):380--392, 2007.

\end{thebibliography}

\end{document}